\documentclass{amsart}
\usepackage{amsmath, amsthm, amscd, amsfonts, amssymb, color}
\usepackage[dvips]{graphicx}
\usepackage{longtable}
\usepackage{tikz}
\usepackage[usenames,dvipsnames]{pstricks}
\usepackage{epsfig}

\newtheorem{thm}{Theorem}

\newtheorem{defn}{Definition}
\newtheorem{rem}{\bf{Remark}}

\newtheorem{prethm}{{\bf Theorem}}

\begin{document}

\title{AN ADIABATIC QUANTUM ALGORITHM FOR DETERMINING GRACEFULNESS OF A GRAPH}

\author{Sayed Mohammad Hosseini$^{*}$, Mahdi Davoudi Darareh, \\ Shahrooz Janbaz and Ali Zaghian}
\thanks{$^{*}$ Corresponding author: hoseinism@mut-es.ac.ir \\
\scriptsize Faculty of Applied Sciences, Malek-e-Ashtar University of Technology, P.O.Box 115/83145, Esfahan, Iran.}


\maketitle
\pagestyle{myheadings}
\markboth{S.M. HOSSEINI, M. DAVOUDI, S. JANBAZ, A. ZAGHIAN}{AN AQC METHOD FOR DETERMINING GRACEFULNESS OF A GRAPH}

\begin{abstract}
Graph labelling is one of the noticed contexts in combinatorics and graph theory. Graceful labelling for a graph $G$ with $e$ edges, is to label the vertices of $G$ with $0, 1, \cdots, e$ such that, if we specify to each edge the difference value between its two ends, then any of $1, 2, \cdots, e$ appears exactly once as an edge label. For a given graph, there is still few efficient classical algorithms that determines either it is graceful or not, even for trees - as a well-known class of graphs. In this paper, we introduce an adiabatic quantum algorithm, which for a graceful graph $G$ finds a graceful labelling. Also, this algorithm can determine if $G$ is not graceful. Numerical simulations of the algorithm reveal that its time complexity has a polynomial behaviour with the problem size up to the range of 15 qubits. 
\end{abstract}

\noindent \small {{\bf keywords:} Graceful Labelling; Adiabatic Quantum Computation; Quantum Algorithm; Combinatorial Optimization Problem.}

\section{Introduction}	
\label{intro}
Based on the adiabatic theorem\cite{book}, adiabatic quantum computation (AQC) is a quantum algorithm which was introduced in 1998, first as an alternative for Grover's search algorithm.\cite{GrovFarhi} AQC is shown to be robust against unitary control errors and decoherence\cite{Robustness}, thus, it might be simpler for experimental implementations.\cite{nature} In fact, AQC finds the ground state of a predetermined problem Hamiltonian. Therefore, an algorithm based on AQC is inherently suitable for solving optimization problems.\cite{ThesisF,QEFarhi} In general, to solve a problem by AQC, we have to formulate the problem as an optimization problem, which the optimal value for its cost function is zero.

To set up an AQC system, first we should determine the initial Hamiltonian $H_{0}$, which is defined to have a known and easy-to-construct ground state. The system will be initiated in the ground state of $H_{0}$. Second, the problem Hamiltonian $H_{p}$ must be determined, which has all possible values of the total cost function as its eigenvalues. The ground state of $H_{p}$ represents the solution of the problem. Next, the system is set up in the ground state of $H_{0}$, and is evolved to the ground state of $H_{p}$ by a general interpolating scheme:
\begin{equation}
\centering
H(t)=[1-s(t)]H_{0}+s(t)H_{p},\label{Adiapath}
\end{equation}
slowly enough to fulfil the adiabatic conditions\cite{ThesisF}, where the function $s(t)$ varies from $0$ to $1$. According to the adiabatic theorem, the system remains in the ground state of its instantaneous Hamiltonian (\ref{Adiapath}) during the total evolution time $T$. The total evolution time $T$, must be determined proportional to inverse square of $g_{min}$, the minimum gap between the two lowest energy levels during the whole evolution\cite{book,QCFarhi}. Time complexity of an adiabatic algorithm is usually determined by analysing how $g_{min}$ (and consequently $T$) changes with the problem size.

It is important to notice that, quantum methods are inherently probabilistic, so they are usually set up to be used with a predetermined success probability. Generally, an important factor in determining this probability is the coherence time, the time after which the system undergoes decoherence and the algorithm fails. However, as we mentioned above, for an adiabatic algorithms the latter seems to be not the case. But for some reasons such as possible degeneracy in the spectrum of $H(t)$, which causes $g_{min}$ to vanish, $T$ cannot be determined finitely. So, we need to consider a finite alternative for the total evolution time such as $t=T^{\prime}$, for which the system may experience excitation and the algorithm may fail. Thus, we still need to consider a success probability $P_{s}$. Here $P_{s}$ is the probability of finding the system in the ground state of $H_{p}$, after a measurement done at $t=T^{\prime}$.
  We describe this constraint further and show how to define $P_{s}$ in the presence of such degeneracies and how to estimate $T^{\prime}$ using $P_{s}$. 
 
  In the first 2000's, the pros and cons of AQC were discussed on some instances of satisfiability problem (SAT).\cite{QCFarhi,VanDam} The efficiency of the method and limitations on application are also discussed consequently by the pioneers and others in some different works.\cite{Fail,NG,NPCV}. However, several AQC algorithms have been suggested to solve different hard problems and some of them have been experimentally implemented. SAT\cite{2satExp,3satExp}, integer factorization\cite{IntFactor}, Simon's problem\cite{Simon,SimonExp}, and some problems in graph theory such as Ramsey number of a graph\cite{ramsey,ramseyExp}, isomorphism\cite{isomorphism} and travelling salesman problem (TSP)\cite{TSP}, are recently studied (see also Refs.~\cite{nature,ThesisF}). A wide variety of other NP and hard problems in mathematics can be converted into a kind of optimization problem, called combinatorial optimization problem (COP). So, there are many potential candidate problems for being solved by an AQC algorithm. In this paper, we want to take the advantage of this property to introduce an adiabatic quantum algorithm for the graceful labelling problem, which obviously belongs to the complexity class NP.
 
 For the first time, Rosa\cite{Rosa} called a function $f$ a $\beta$-\textit{valuation of a graph} $G$ with $m$ edges, if $f$ is an injection from the vertices of $G$ to the set $\lbrace 0,1,\ldots,m\rbrace$ such that, when each edge $\lbrace v,w\rbrace$ is assigned the label $\vert f(v)-f(w)\vert$, the resulting edge labels are distinct. Golomb\cite{Golomb} subsequently called such labellings \textit{graceful labellings} and this is now the popular term. Erd\H os had believed that almost all graphs are not graceful, but many graphs that have some sort of regularity of structure are graceful.\cite{Erdos} Also, Sheppard\cite{Sheppard} has shown that there are exactly $e!$ gracefully labelled graphs with $e$ edges.

Graceful graphs have found a wide range of applications in different fields of science, such as  the X-ray crystallographic analysis, coding theory, communication network addressing, optimal circuit design, and database management.\cite{Bloom}

 One of the most famous and long standing open problems in graceful labelling is conjectured by Ringel and Kotzig, known as graceful tree conjecture\cite{Ringel-Kotzig}. This conjecture states that ``all trees are graceful", and it is verified just for all trees up to $35$ vertices. Also, we know that every tree can be embedded as an induced subgraph of a graceful tree.\cite{InTree} So, there is no forbidden subgraph characterization of various particular kinds of graceful graphs. Graceful trees have some further applications in combinatorial problems.\cite{stateofart}

The graceful labelling problem, by now, is solved only for very special cases of graphs. There are few optimization methods that allow one to determine the gracefulness of a given graph, in general. These methods are two mathematical programming methods\cite{Redl,Eshghi} and one meta-heuristic method based on ant colony optimization.\cite{Meta} For more details and results one can see Ref.~\cite{Galian}, where it is an extensive survey on graph labellings that is periodically updated.

In this paper, we show that our adiabatic quantum algorithm is generally applicable for all graphs. Though, the most efficiency is for the trees, as well as any disconnected graph in which the number of vertices is exactly one more than the number of edges. In these cases, the system needs the least number of qubits.  Our approach implies that the problem can be implemented on much faster devices than that are already presented, with probably much less resources.

The paper is organized as follows: after the introduction, Sect.~\ref{basic} presents some basic definitions, concepts and one theorem that enables us to formulate the graceful labelling as an optimization problem, what we will do in Sect.~\ref{cop}. Sect.~\ref{aqc} explains the adiabatic quantum formalism of our algorithm, and Sect.~\ref{impl} gives some notes for its implementation. Finally, we analyse our simulation results in Sect.~\ref{sim_res}, and conclude in the last section.

\section{Preliminaries}
\label{basic}
In the following, we depict the mathematical backgrounds we need to describe our algorithm. 
\begin{defn}\label{GL}
A graceful labelling of a graph $G=(V,E)$ with $N$ vertices and $e$ edges is a one-to-one mapping $\Psi$ of the vertex set $V(G)=\lbrace 0,1,...,N-1 \rbrace$ into the set $\lbrace 0,1,...,e\rbrace$ with the following property: If we define, for any edge $e_{i}=\lbrace u,v\rbrace \in E(G)$, the value $\Psi^{\bullet}(e_{i})=\vert \Psi(u)-\Psi(v)\vert$, then $\Psi^{\bullet}$ is a one-to-one mapping of the set $E(G)$ onto the set $\lbrace1,2,...,e\rbrace$. We define $L_{V}(G)$ and $L_{E}(G)$, the set of vertex and edge labels of $G$, respectively.
\end{defn}
Therefore, when $G$ admits a graceful labelling, $L_{V}(G)\subseteq\lbrace0,1,...,e\rbrace$ and $L_{E}(G)=\lbrace1,2,...,e\rbrace$.  
A graph which admits a graceful labelling is called a graceful graph.

A graph with $N$ vertices can have from $0$ to $\frac{N(N-1)}{2}$ edges. Clearly by definition, a graph with $e+1<N$, cannot admit a graceful labelling, since according to Rosa\cite{Rosa} ``it has too many vertices" and ``not enough labels" to be graceful. So we can focus only on graphs with $e+1\geqslant N$.
\begin{defn} 
Let $G$ be a graph with $N$ vertices. The adjacency matrix of $G$ is an $N\times N$ matrix $A$ with elements $ a_{ij}$ such that

\centering
$a{}_{ij}=\begin{cases} 
1\   $   if vertices $v_{i}$ and $v_{j}$ are connected $& \\0\     $   o.w.$ \end{cases}$.
\end{defn} 

Clearly, when the graph $G$ is simple, its adjacency matrix $A$ is symmetric. Moreover, the number of ones above/under its main diagonal equals to the number of edges in the graph $G$. 

\begin{figure}[htbp]
\centering
\usetikzlibrary{shapes.geometric}
\begin{tikzpicture}
[every node/.style={inner sep=0pt}]
\node (1) [circle, minimum size=16.25pt, fill=white, line width=0.625pt, draw=black] at (62.5pt, -137.5pt) {\textcolor{black}{1}};
\node (4) [circle, minimum size=16.25pt, fill=white, line width=0.625pt, draw=black] at (112.5pt, -137.5pt) {\textcolor{black}{4}};
\node (5) [circle, minimum size=16.25pt, fill=white, line width=0.625pt, draw=black] at (162.5pt, -137.5pt) {\textcolor{black}{5}};
\node (3) [circle, minimum size=16.25pt, fill=white, line width=0.625pt, draw=black] at (212.5pt, -137.5pt) {\textcolor{black}{3}};
\node (0) [circle, minimum size=16.25pt, fill=white, line width=0.625pt, draw=black] at (137.5pt, -100.0pt) {\textcolor{black}{0}};
\draw [line width=0.625, color=black] (1) to  (4);
\draw [line width=0.625, color=black] (4) to  (5);
\draw [line width=0.625, color=black] (5) to  (3);
\draw [line width=0.625, color=black] (5) to  (0);
\draw [line width=0.625, color=black] (4) to  (0);
\node at (87.5pt, -130.0pt) {\textcolor{red}{3}};
\node at (137.5pt, -130.0pt) {\textcolor{red}{1}};
\node at (187.5pt, -130.0pt) {\textcolor{red}{2}};
\node at (154.375pt, -113.75pt) {\textcolor{red}{5}};
\node at (120.0pt, -114.375pt) {\textcolor{red}{4}};
\end{tikzpicture}
\epsfig{file=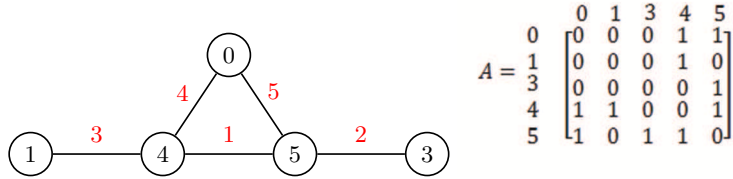, scale=1}
\vspace*{8pt}
\caption{A gracefully labelled graph $G$, and its adjacency matrix $A$.}
\label{SampleG}
\end{figure}

\begin{defn} 
A permutation $\pi$ of a finite set $S=\lbrace0,1,...,n\rbrace$ is a one-to-one correspondence on the set $S$, which sends $i\rightarrow \pi_{i}$ such that $\pi_{i}\in S$, and $\pi_{i}\neq\pi_{j}$ for $i\neq j$. \end{defn} 

We know that any permutation on such set with $n+1$ elements, can be represented by an $(n+1)\times(n+1)$ unitary matrix $P_{\pi}$, which is the permutation $\pi$ applied to the rows of the identity matrix $I_{n+1}$. Another representation for the permutation $\pi$ is
\begin{equation}\label{2rowPi} 
\centering
\pi=\begin{pmatrix}
0 & \cdots & i & \cdots & n \\ 
\pi_0 & \cdots & \pi_i & \cdots & \pi_n
\end{pmatrix},
\end{equation}
 \noindent where column $i$ indicates that $\pi$ sends $i\rightarrow\pi_{i}$ . The elements of the matrix $P_{\pi}$, therefore, are binary numbers given by $[P_{\pi}]_{ij}=\delta_{\pi_{i},j}$,
  where $\delta_{x,y}$  is the Kronecker delta of $x$ and $y$. For simplicity, from now on, we denote $P_{\pi}$ by $P$.
\begin{defn}\label{Ext}
Suppose $G$ is a graph with $N$ vertices and $e$ edges. The extension of the graph $G$, which is denoted by $G^{\prime}$, is the union of the graph $G$ with $r=e+1-N$ isolated vertices. \end{defn}
So $G^{\prime}$ has $N+r=e+1$ vertices, which are labelled from $0$ to $N^{\prime}$, and $N^{\prime}=e$.
Obviously, $G^{\prime}$ is graceful whenever $G$ admits a graceful labelling.
Similarly one can define $A^{\prime}$ to be the extension of matrix $A$, which is the $(N^{\prime}+1)\times (N^{\prime}+1)$ adjacency matrix of $G^{\prime}$, that has just some zero rows and columns more than $A$. We can see that the total number of 1's in both  $A$ and $A^{\prime}$ are equal.

\vspace*{12pt}
\noindent
\begin{rem}\label{nt1} The map $\Psi(A)$ that labels the graph $G$ is an injection ($V(G)\subseteq L_{V}(G)$), while the map $\Psi(A^{\prime})$ that labels the graph $G^{\prime}$ is a bijection ($V(G^{\prime})=L_{V}(G^{\prime})$). This is why we extend a graph like $G$ to its extension $G^{\prime}$. Furthermore, since there are always $r$ elements in $L_{V}(G)$ which $\Psi(A)$ maps no vertices of $G$ to them, we can simply extend any mapping $\Psi(A)$ to a mapping for $G^{\prime}$, by sending any of $r$ added vertices to these $r$ elements arbitrarily and one by one. Therefore, the fact that any labelling $\Psi(A^{\prime})$ is a bijection enables us to define it in terms of $A^{\prime}$ and a permutation matrix, as follows:
\begin{equation}\label{A''} 
\Psi (A^{\prime})= P A^{\prime} P^{T}.
\end{equation}
Here $P$ is a permutation matrix that permutes the rows of $A^{\prime}$, and so $P^{T}$ permutes the columns. Hence, the matrix $\Psi(A^{\prime})$, which we denote it by $A^{\prime\prime}$, is an adjacency matrix of the graph $G^{\prime\prime}$, that is clearly isomorphic\footnote{Two graphs $G_{1}$ and $G_{2}$, respectively with adjacency matrices $A_{1}$ and $A_{2}$, are called isomorphic, if there exist a permutation matrix $P$, such that $A_{2}=P A_{1} P^{T}$.}  to $G^{\prime}$.
\end{rem}

\begin{rem}\label{nt2} In an extended adjacency matrix $A^{\prime}$, the equality $a_{ij}=1$ has two meanings: 1) there exists an edge $e_{k}$ between vertices $i$ and $j$, 2) the edge's label is $\Psi^{\bullet}(e_{k})=|i-j|$. 
\end{rem}

\begin{defn}
Let $A^{\prime}$ be an $(N^{\prime}+1)\times(N^{\prime}+1)$ matrix. The minor diagonal $b_{i}$, $1\leq i\leq N^{\prime}$, is the following sequence of the elements: 
\begin{center}
$b_{i}: a^{\prime}_{0 i}, a^{\prime}_{1 i+1},\cdots, a^{\prime}_{j(i+j)}, \cdots, a^{\prime}_{(N^{\prime}-i)N^{\prime}}$,
\end{center}
where the length of the sequence $b_{i}$ is $N^{\prime}-i+1$.
\end{defn}

\noindent We can see that each $b_{i}$ is a sequence of elements of $A^{\prime}$, located parallel to the main diagonal.
The last two definitions are illustrated in Fig.~\ref{Extensions}, based on the graph $G$ in Fig.~\ref{SampleG}.

\begin{figure}[htbp]
\centering
\begin{tikzpicture}
[every node/.style={inner sep=0pt}]
\node (1) [circle, minimum size=16.25pt, fill=white, line width=0.625pt, draw=black] at (62.5pt, -137.5pt) {\textcolor{black}{1}};
\node (4) [circle, minimum size=16.25pt, fill=white, line width=0.625pt, draw=black] at (112.5pt, -137.5pt) {\textcolor{black}{4}};
\node (5) [circle, minimum size=16.25pt, fill=white, line width=0.625pt, draw=black] at (162.5pt, -137.5pt) {\textcolor{black}{5}};
\node (3) [circle, minimum size=16.25pt, fill=white, line width=0.625pt, draw=black] at (212.5pt, -137.5pt) {\textcolor{black}{3}};
\node (0) [circle, minimum size=16.25pt, fill=white, line width=0.625pt, draw=black] at (137.5pt, -100.0pt) {\textcolor{black}{0}};
\node (2) [circle, minimum size=16.25pt, fill=white, line width=0.625pt, draw=black] at (75.0pt, -100.0pt) {\textcolor{black}{2}};
\draw [line width=0.625, color=black] (1) to  (4);
\draw [line width=0.625, color=black] (4) to  (5);
\draw [line width=0.625, color=black] (5) to  (3);
\draw [line width=0.625, color=black] (5) to  (0);
\draw [line width=0.625, color=black] (4) to  (0);
\node at (87.5pt, -130.0pt) {\textcolor{red}{3}};
\node at (137.5pt, -130.0pt) {\textcolor{red}{1}};
\node at (187.5pt, -130.0pt) {\textcolor{red}{2}};
\node at (154.375pt, -113.75pt) {\textcolor{red}{5}};
\node at (120.0pt, -114.375pt) {\textcolor{red}{4}};
\end{tikzpicture}
\epsfig{file=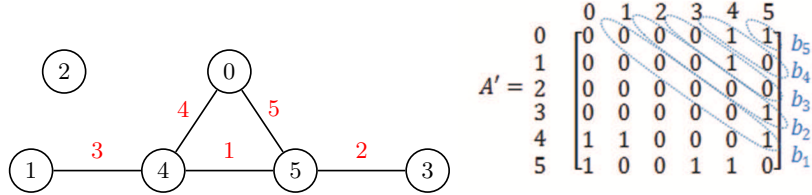, scale=1}
\vspace*{8pt}
\caption{The graph $G^{\prime}$ and its adjacency matrix $A^{\prime}$, the extensions of $G$ and $A$ by adding r=1 isolate vertex to $G$. Minor diagonals are circuited.}
\label{Extensions}
\end{figure}

\begin{rem}\label{nt3} Since $b_{i}$ is a binary sequence, we consider $m(b_{i})$ to be its Hamming weight \footnote{The Hamming weight of a binary string is the number of its non-zero elements.}. The important gain of extending a graph is that in an extended adjacency matrix $A^{\prime}$, for all elements of each minor diagonal $b_{i}$, the difference between the row and column indices is fixed and equal to $i$ (see Rem.~\ref{nt2}). Therefore $m(b_{i})$ is the number of edges with label $i$. Consequently we have
\begin{equation}\label{note3}
\sum_{i=1}^{N^{\prime}}m(b_{i})=e=N^{\prime}.
\end{equation}
\end{rem}
Now we are ready to state our main idea. 
\begin{thm}\label{Thm}
Let $G^{\prime}$ be the extension of the graph $G$ with the adjacency matrix $A^{\prime}$ with dimension $N^{\prime}+1$. Then $G$ is graceful if and only if in $A^{\prime}$ we have  $m(b_{i})=1$, for all $1\leq i\leq N^{\prime}$.
\end{thm} 
\begin{proof}
First consider the graph $G$ is graceful. By Definition \ref{GL}, we have $L_{E}(G)=\lbrace1,2,...,e\rbrace=L_{E}(G^{\prime})$, which means that for each $1\leq i\leq N^{\prime}$ there exist at least one edge with label $i$. According to Rem.~\ref{nt3}, it means that each $b_{i}$ contains at least one 1, that is
\begin{equation}\label{prf1}
\centering
 \ m(b_{i})\geq 1 \; ;\; \; i=1,2,\ldots ,N^{\prime}.
\end{equation}
On the other hand, the Definition \ref{Ext} together with equality (\ref{note3}) and inequality (\ref{prf1}), imply that each $b_{i}$ has exactly one 1, which means $m(b_{i})=1$. Conversely, if for all $1\leq i\leq N$ we have $m(b_{i})=1$, it means that $G^{\prime}$ has exactly one edge with label $i$ for $1\leq i\leq N^{\prime}$. Thus, we have $L_{E}(G)=L_{E}(G^{\prime})=\lbrace1,2,...,e\rbrace$ which means that $G$ is gracefully labelled. This completes our proof. 
\end{proof}

\section{Associated Combinatorial Optimization Problem (COP)}
\label{cop}
In this section we introduce a COP which is equivalent to finding a graceful labelling for a graph $G$. The search space is the Hamming space of binary strings like $s_{b}$ of length $(N^{\prime}+1)U$ bits, where $N^{\prime}=e$ is  the number of edges in $G$ and $U$ is the minimum number of bits required for binary representation of $N^{\prime}$, i.e. $U=\lceil Log_{2}(N^{\prime}+1)\rceil$. The cost function $C(s_{b})$ is a positive valued function which is minimized for the optimal bit string $s^{*}_{b}$, i.e. $C(s^{*}_{b})=0$.

As the first step, we note that our problem is to find a mapping that relabels $G^{\prime}$ to a graceful one. Rem.~\ref{nt1} states that how such a bijection (on $L_{V}(G^{\prime})$) can be rewritten in terms of $A^{\prime}$ (which is fixed for each graph) and a permutation matrix $P$ (which differs for different mappings). So the problem of finding a graceful labelling for a given graph $G$, can be considered as the problem of finding a permutation matrix that leads to the desired labelling for its extension $G^{\prime}$.
We can uniquely correspond an $(N^{\prime}+1)\times(N^{\prime}+1)$ permutation matrix $P$ to an $(N^{\prime}+1)U$-bit binary string. The idea is simply as follows and fully described in Ref.~\cite{isomorphism}.

Considering the other representation of $\pi$ in (\ref{2rowPi}), it is clear that the bottom row contains all information about $\pi$. Therefore, one can map a permutation to the integer sequence $s_{int}:=\pi_{0},\cdots,\pi_{i},\cdots,\pi_{n}$, or equivalently, to the sequence of their $U$-bit binary representation:

\begin{equation}\label{binPiseq}
\centering
s_{b}=\underbrace{(s_{0},s_{1},\cdots ,s_{U-1})}_{\text{$\pi_{0}$}},\underbrace{(s_{U},s_{U+1},\cdots ,s_{2U-1})}_{\text{$\pi_{1}$}},\cdots ,\underbrace{(s_{nU},s_{nU+1},\cdots ,s_{(n+1)U-1})}_{\text{$\pi_{n}$}}.
\end{equation}
Thus, we uniquely mapped a permutation matrix to an $(n+1)U$-bit binary string.

Conversely, not all binary strings correspond to a permutation. In fact, to map an $(n+1)U$-bit binary string $s_{b}$ to a permutation, its corresponding integer string needs to contain each element of $S$ just once. Explicitly, the following conditions must be held:
\begin{list}{}{}
\item[(i)] $s_{int}$ should not contain any integer larger than $n$,  
\item[(ii)] $s_{int}$ should not have repetition.
\end{list} 
Since our approach is based on the extended adjacency matrix of the graph $G$, from now on, we consider the $(N^{\prime}+1)\times(N^{\prime}+1)$ permutation matrices, i.e. $n=N^{\prime}$.

For the next step, we determine the structure of $C(s_{b})$. As mentioned above, the string $s_{b}$ must satisfy the conditions (i) and (ii). Also, we need a condition to guarantee that a string $s_{b}$ corresponds to a permutation that leads to a graceful labelling. The idea for this condition is obtained from Theorem~\ref{Thm}. We define the total cost function as follows: 
\begin{equation}\label{Total C}
\centering
C(s_{b})=C_{1}(s_{b})+C_{2}(s_{b})+C_{3}(s_{b}),
\end{equation}
where $C_{i}(s_{b})$, $i=1,2,3$ are non-negative cost functions for the conditions above. Thus, $C(s_{b})$ vanishes if and only if all $C_{i}$'s tend to zero, i.e. $s_{b}$ satisfies all conditions. The cost functions $C_{1}$ and $C_{2}$ must guarantee that $s_{b}$ corresponds to a permutation matrix $P$, so according to (i) and (ii), we can write
\begin{equation}\label{C1}
\centering
C_{1}(s_{b})=\sum_{i=0}^{N^{\prime}}\sum_{k=N^{\prime}+1}^{M^{\prime}}\delta_{\pi_{i},k},
\end{equation}

\begin{equation}\label{C2}
\centering
C_{2}(s_{b})=\sum_{i=0}^{N^{\prime}-1}\sum_{j=i+1}^{N^{\prime}}\delta_{\pi_{i},\pi_{j}},
\end{equation}
where $M^{\prime}=2^{U}-1$ is the largest integer that can be represented by $U$ bits.

 If we consider $U$-bit binary representations of $x$ and $y$, i.e. $x=x_{0}, \cdots, x_{U-1}$ and $y=y_{0}, \cdots, y_{U-1}$, then

\begin{equation}\label{IntDelt2Bin} 
\centering
\delta_{x,y}=\prod_{i=0}^{U-1}\delta_{x_{i},y_{i}}=\prod_{i=0}^{U-1}(x_{i}+y_{i}-1)^{2}=\begin{cases} 
1\ \quad x_{i}=y_{i}\,; \, \forall i \\ 0\ \quad x_{i}\neq y_{i}\,;  $    \textit{for some i}$  \end{cases}.
\end{equation}
So we have

\begin{equation} \label{C1'}
\centering
C_{1}(s_{b})=\sum_{i=0}^{N^{\prime}}\sum_{k=N^{\prime}+1}^{M^{\prime}}\prod_{r=0}^{U-1}(s_{iU+r}+k_{r}-1)^{2},
\end{equation}
where $k_{r}$ is the $r$-th bit in the binary representation of integer $k$; and

\begin{equation}\label{C2'}
\centering
C_{2}(s_{b})=\sum_{i=0}^{N^{\prime}-1}\sum_{j=i+1}^{N^{\prime}}\prod_{r=0}^{U-1}(s_{iU+r}+s_{jU+r}-1)^{2}.
\end{equation}

\noindent Now, we deduce the third cost function from Theorem~\ref{Thm} as follows:

\begin{equation}\label{C3}
\centering
C_{3}(s_{int})=\sum_{i=1}^{e}\left[ 1-\sum_{k=0}^{e-i}a^{\prime\prime}_{k,(i+k)}\right] ^{2},
\end{equation}
in which $a^{\prime\prime}_{k,(i+k)}$ is the $k$-th element of $b_{i}$, which is the $i$-th minor diagonal of $A^{\prime\prime}$.

The last step is to write $C_{3}(s_{int})$ explicitly in terms of the bits of $s_{b}$. According to (\ref{A''}) we have $A^{\prime\prime}=P A^{\prime} P^{T}$. Therefore, since each element of $b_{i}$ is in fact an element of $A^{\prime\prime}$, we can rewrite them in terms of the elements of $P$ and $A^{\prime}$:
\begin{equation}\label{a"ij}
 \centering
a^{\prime\prime}_{ij}=\sum_{k=0}^{N^{\prime}}\sum_{r=0}^{N^{\prime}}p_{ir} a^{\prime}_{rk}p_{kj}^{T}=\sum_{k=0}^{N^{\prime}}\sum_{r=0}^{N^{\prime}}p_{ir} a^{\prime}_{rk}p_{jk}.
\end{equation}
\noindent Finally, we should write each $p_{ij}$ in terms of binary elements of $s_{b}$:
\begin{equation}
\centering
p_{ij}=\left[ P_{\pi} \right]_{ij}=\delta_{i,\pi_{j}}=\prod_{r=0}^{U-1}(i_{r}+s_{jU+r}-1)^{2},\label{sigma ij}
\end{equation}
where $i_{r}$ is the $r$-th bit in the binary representation of integer $i$. By substituting (\ref{sigma ij}) in (\ref{a"ij}), and using (\ref{C3}), we obtain $C_{3}(s_{b})$.

\section{Adiabatic quantum computation for gracefulness of a graph}
\label{aqc}
Now we are ready to introduce an adiabatic quantum algorithm for graceful labelling problem. According to (\ref{Adiapath}), we just need to identify $H_{0}$ and $H_{p}$ for our specific problem. To do so, we map the corresponding COP onto an adiabatic quantum computation model. Consider the Hamming space of binary strings like $s_{b}$ of length $L=(N^{\prime}+1)U$ bits. Since this approach promotes each bit in $s_{b}$ to a qubit, our quantum register would also contain $L$ qubits. We assume the Hilbert space of the quantum register to be the span of the computational basis states (CBS) $\vert s_{b}\rangle$, which are the $2^{L}$ eigenstates of $\sigma_{z}^{0}\otimes \cdots \otimes \sigma_{z}^{L-1}$, and $\sigma_{z}^{l}$ is the $z$-Pauli operator corresponding to the $l$-th bit of $s_{b}$. Now, the problem Hamiltonian $H_{p}$ is defined to be diagonal in the CBS, with the eigenvalue $C(s_{b})$ for each eigenstate $\vert s_{b}\rangle$:
\begin{equation}\label{Hp_0}
H_{p}\vert s_{b}\rangle =C(s_{b}) \vert s_{b}\rangle \ \ ;\ \ s_{b}\in {\lbrace 0,1\rbrace}^{2^{L}}.
\end{equation}
This necessitates that: 
\begin{equation}\label{Hp}
\centering
H_{p}=\sum_{s_{b}}C(s_{b})\vert s_{b} \rangle\langle s_{b} \vert ,
\end{equation}
in which the $l$-th qubit of the quantum register ($0\leq l \leq L-1$), is described by the one-bit Hamiltonian $\frac{1}{2}[I^{l}-\sigma^{l}_{z}]$, where $I^{l}$ is the two dimensional identity operator corresponding to this qubit\cite{QCFarhi}. We can see that the eigenstates of $H_{p}$ correspond to all possible bit strings $s_{b}$.

On the other hand, since the initial Hamiltonian $H_{0}$ should not be diagonal in the basis that diagonalizes $H_{p}$, we choose:
\begin{equation}\label{H0}
\centering
H_{0}=\sum_{l=0}^{L-1}\frac{1}{2}[I^{l}-\sigma ^{l}_{x}],
\end{equation}
as a well-known initial Hamiltonian for COPs, where $\sigma^{l}_{x}$ is the x-Pauli operator for qubit $l$. Clearly, the ground state of $H_{0}$ (in the basis of eigenstates of $H_{p}$) is the easy-to-construct uniform superposition of all CBS\cite{QCFarhi}.

Finally, by substituting Eqs. (\ref{Hp}) and (\ref{H0}) in Eq. (\ref{Adiapath}), we obtain the total time-dependent Hamiltonian $H(t)$ for the adiabatic evolution that determines the gracefulness of $G$ at time $T$, where $s(T)=1$. 

\section{Illustrative Notes for Implementation}
\label{impl}
 To illustrate our adiabatic approach, we consider the only 3-vertex tree, $K_{1,2}$, as an input for our algorithm (see Fig.~\ref{P3P}).
 
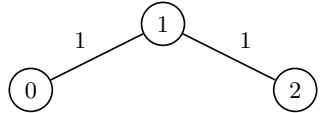
\begin{figure}[htbp]
\centerline{\begin{tikzpicture}
[every node/.style={inner sep=0pt}]
\node (2) [circle, minimum size=16.25pt, fill=white, line width=0.625pt, draw=black] at (112.5pt, -100.0pt) {\textcolor{black}{1}};
\node (1) [circle, minimum size=16.25pt, fill=white, line width=0.625pt, draw=black] at (62.5pt, -125.0pt) {\textcolor{black}{0}};
\node (3) [circle, minimum size=16.25pt, fill=white, line width=0.625pt, draw=black] at (162.5pt, -125.0pt) {\textcolor{black}{2}};
\draw [line width=0.625, color=black] (2) to  (1);
\draw [line width=0.625, color=black] (3) to  (2);
\node at (81.25pt, -106.25pt) {\textcolor{black}{1}};
\node at (143.75pt, -106.25pt) {\textcolor{black}{1}};
\end{tikzpicture}}
\vspace*{8pt}
\caption{Input labelling for our AQC algorithm for $K_{1,2}$.}
\label{P3P}
\end{figure}

Our adiabatic algorithm will search for permutations that change the input labelling to a graceful one. When the input of the algorithm is a tree with $e$ edges we do not need to extend the graph, since we have $N=e+1$ and $r=0$. Thus, for $K_{1,2}$ we have $N^{\prime}=e=2 $ and $ U=2$. Therefore, we obtain $M^{\prime}=3$ and  $L=6$. So we need a system with 6 qubit, prepared in the ground state of (\ref{H0}), which after obeying an adiabatic evolution (\ref{Adiapath}) with $s(t)=\dfrac{t}{T}$, will end up in the ground state of 

\begin{equation}\label{HpP3}
H_p=\dfrac{1}{16} (48-6\sigma^{2}_{z}-8\sigma^{3}_{z}+6\sigma^{2}_{z}\sigma^{3}_{z}-5\sigma^{4}_{z}+5\sigma^{2}_{z}\sigma^{4}_{z}-\sigma^{3}_{z}\sigma^{4}_{z}+\sigma^{2}_{z}\sigma^{3}_{z}\sigma^{4}_{z}-4\sigma^{5}_{z}+4\sigma^{3}_{z}\sigma^{5}_{z}$$$$+5\sigma^{4}_{z}\sigma^{5}_{z}+\sigma^{2}_{z}\sigma^{4}_{z}\sigma^{5}_{z}+\sigma^{3}_{z}\sigma^{4}_{z}\sigma^{5}_{z}+5\sigma^{2}_{z}\sigma^{3}_{z}\sigma^{4}_{z}\sigma^{5}_{z}-\sigma^{1}_{z}(4-8\sigma^{5}_{z}-2\sigma^{2}_{z}\sigma^{5}_{z}+\sigma^{4}_{z}(1+\sigma^{2}_{z}-\sigma^{5}_{z}+$$$$ \sigma^{2}_{z}\sigma^{5}_{z})+\sigma^{3}_{z}(-4+2\sigma^{2}_{z}\sigma^{5}_{z}+\sigma^{4}_{z}(1+\sigma^{2}_{z}-\sigma^{5}_{z}+\sigma^{2}_{z}\sigma^{5}_{z})))+\sigma^{0}_{z}(-5-\sigma^{3}_{z}+6\sigma^{4}_{z}+2\sigma^{3}_{z}\sigma^{4}_{z}-\sigma^{5}_{z}-$$$$
\sigma^{3}_{z}\sigma^{5}_{z}+\sigma^{2}_{z}(5+\sigma^{3}_{z}-\sigma^{5}_{z}-\sigma^{3}_{z}\sigma^{5}_{z})+\sigma^{1}_{z}(5+\sigma^{5}_{z}+6\sigma^{4}_{z}\sigma^{5}_{z}-\sigma^{2}_{z}(-1+\sigma^{3}_{z}(-5+\sigma^{5}_{z})+\sigma^{5}_{z})+$$$$
\sigma^{3}_{z}(1+\sigma^{5}_{z}+2\sigma^{4}_{z}\sigma^{5}_{z})))),
\end{equation}
 at $t=T$. To obtain (\ref{HpP3}), we replaced each bit of $s_{b}$ in (\ref{Total C}) by the operator $\frac{1}{2}[I^{l}-\sigma ^{l}_{z}]$ for $0\leq l \leq L-1$.
The evolution of some first eigenvalues (corresponding to the instantaneous lowest energy levels) of $H(t)$ during the procedure is shown in Fig.~\ref{P3Elevels}.  
\begin{figure}[htbp]
\includegraphics[scale=1]{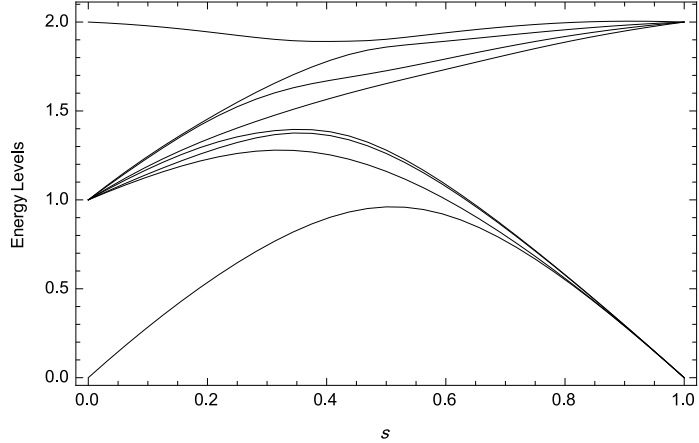} 
\vspace*{8pt}
\caption{The first eight (out of 64) energy levels of $H(s)$ during the adiabatic evolution for finding a graceful labelling for $K_{1,2}$.}
\label{P3Elevels}       
\end{figure}

One can see that for this $H_{p}$, the degree of degeneracy, $D$, is four. So, after a true adiabatic evolution, the quantum register would be in a superposition of the following degenerated ground states of $H_{p}$
\begin{equation}\label{AnsKets}
\lbrace \vert 001001\rangle ,\vert 011000\rangle ,\vert 010010\rangle ,\vert 100001\rangle \rbrace .
\end{equation}
Thus, after measurement we may obtain any of the corresponding bit strings or their equivalent integer strings $(0,2,1)$, $(1,2,0)$, $(1,0,2)$, or $(2,0,1)$. The corresponding permutations ($\pi^{1}$, $\pi^{2}$, $\pi^{3}$ and $\pi^{4}$, respectively) change the current labelling to graceful ones, which are the only possible graceful labellings for $K_{1,2}$, divided into two isomorphic classes:
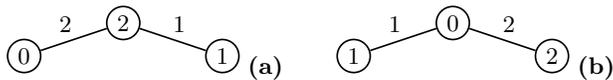
\begin{figure*}[htbp]
\centerline{\begin{tikzpicture}
[every node/.style={inner sep=0pt}]
\node (1) [circle, minimum size=12.5pt, fill=white, line width=0.625pt, draw=black] at (125.0pt, -100.0pt) {\textcolor{black}{1}};
\node (0) [circle, minimum size=12.5pt, fill=white, line width=0.625pt, draw=black] at (50.0pt, -100.0pt) {\textcolor{black}{0}};
\node (2) [circle, minimum size=12.5pt, fill=white, line width=0.625pt, draw=black] at (87.5pt, -87.5pt) {\textcolor{black}{2}};
\draw [line width=0.625, color=black] (0) to  (2);
\draw [line width=0.625, color=black] (2) to  (1);
\node at (65.625pt, -88.125pt) {\textcolor{black}{2}};
\node at (108.75pt, -88.125pt) {\textcolor{black}{1}};
\end{tikzpicture}
\textbf{(a)} \; \; \;
\begin{tikzpicture}
[every node/.style={inner sep=0pt}]
\node (2) [circle, minimum size=12.5pt, fill=white, line width=0.625pt, draw=black] at (125.0pt, -100.0pt) {\textcolor{black}{2}};
\node (1) [circle, minimum size=12.5pt, fill=white, line width=0.625pt, draw=black] at (50.0pt, -100.0pt) {\textcolor{black}{1}};
\node (0) [circle, minimum size=12.5pt, fill=white, line width=0.625pt, draw=black] at (87.5pt, -87.5pt) {\textcolor{black}{0}};
\draw [line width=0.625, color=black] (1) to  (0);
\draw [line width=0.625, color=black] (0) to  (2);
\node at (65.625pt, -88.125pt) {\textcolor{black}{1}};
\node at (108.75pt, -88.125pt) {\textcolor{black}{2}};
\end{tikzpicture}
\textbf{(b)}}
\vspace*{8pt}
\caption{Output labellings for $K_{1,2}$. $\pi^{1}$ and $\pi^{2}$ relabel the graph to labellings isomorphic to (a), and $\pi^{3}$ and $\pi^{4}$ relabel it to labellings isomorphic to (b).}
\label{RlP3}
\end{figure*}

Therefore, degeneracy in the ground state of $H_{p}$ is inherent for this optimization problem. Fig.~\ref{P3Elevels} also shows that the value of the $g_{min}$ tends to zero and the system may experience excitation. This means that the system may no longer be in the ground state of $H_{p}$ at any finite evolution time $t=T^{\prime}$. So we cannot determine the evolution time proportional to inverse square of $g_{min}$, as usually suggested for an adiabatic algorithm in the absence of degeneracy. Generally, it is still unknown that whether an adiabatic evolution with $g_{min}=0$ may end up in a desired state in an acceptable time or not. But fortunately, we can estimate the required time $T^{\prime}$, after which the system will be in one of the degenerated ground states, with a desired success probability $P_{s}$.\cite{QEFarhi,IntFactor} To do so, we solve the Schr\"{o}dinger equation

\begin{equation}\label{Schrod}
\centering
i\frac{d}{dt}\vert\psi(t)\rangle =H(t)\vert \psi(t)\rangle ,    
\end{equation}
for the total Hamiltonian $H(t)$, given in (\ref{intro}). Then we can properly determine the general state of our system at any time $T^{\prime}$, i.e. $\vert \psi(t=T^{\prime}) \rangle$. Consequently, to calculate the probability of the state of the system being the $j$-th degenerated ground state of $H_{p}$, i.e. $\vert \psi_{0}^{j} \rangle_{p}$, we use: 

\begin{equation}\label{PsPart}
\centering
P_{s}^{j}=\vert \langle \psi(t=T^{\prime}) \vert \psi_{0}^{j} \rangle_{p} \vert^{2}  \; ; \; \; 1\leq j\leq D.
\end{equation}

\noindent Then, we define the total success probability, $P_{s}$, as follows:
\begin{equation}\label{Ps}
\centering
P_{s}:=\sum_{j=1}^{D} P_{s}^{j}.
\end{equation}
We call it total success probability since it is a lower bound for the probability of successfully finding the minimum value of the total cost function $C(s_{b})$. As we show in Sect.~\ref{sim_res}, for the current example $K_{1,2}$, we obtained $P_{s}\simeq0.25$ at time $T^{\prime}\simeq2.5$. 

In the next section we see that the graceful labelling problem, even for the simplest example ($K_{1,1}$), has a degree of degeneracy ($D=2$). This actually decreases the probability of the system to pass an adiabatic evolution. However, as we stated above, since any of the degenerated final ground states leads to a graceful labelling, the total success probability for any specified $T^{\prime}$ is (at least) the sum of partial success probabilities ($P_{s}^{j}$) at $t=T^{\prime}$. The importance of this consideration will be more illustrated when we note that for the real-scaled examples, $D$ may increase intensively. Specially, Sheppard in Ref.~\cite{Sheppard} showed that there are exactly $e!$ (non-isomorphic) graceful graphs with $e$ edges, where half of them correspond to different labellings for the same graphs (degenerated answers). Including isomorphic labellings will increase $D$ even more.
 
 We should also note that with the same number of vertices, $D$ for graphs which are not graceful can be much greater than $D$ for graceful graphs (because in this case, $D$ is the number of degenerated ground states with a positive common eigenvalue, that is the number of bit strings which are not necessarily correspond to permutation matrices, but the value of the total cost function is the same for them). 
 This means that, though our algorithm is a true-biased Monte-Carlo algorithm\footnote{A randomized algorithm whose running time is deterministic, but whose output may be incorrect when it returns false, with a certain (typically small) probability.}, when it declares that a graph is not graceful, it can be a reliable output with a high probability. In such cases, by repeating the algorithm, one can ramp up the probability of returning the correct output to a number as close to unity as desired.\cite{MC}

%

\section{Simulation Results}
\label{sim_res}
We now perform a detailed analytical review of our simulation results. Table~\ref{TG} presents a guide to the graphs that we discuss here\footnote{To avoid ambiguity, in this paper we intentionally use $Z_{n}$ to denote paths with $n$ vertices, which are usually denoted by $P_{n}$ in graph theory.}.
\begin{table}
\caption{Guide to the names of the graphs.}
{\begin{tabular}{@{}c@{}}
\includegraphics[scale=1.5]{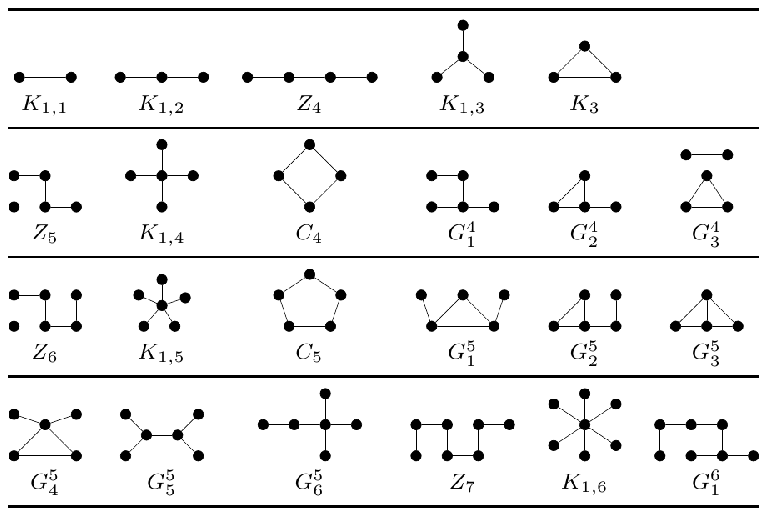}
\end{tabular}}\label{TG}
\end{table}
 We simulated our algorithm for all graphs with up to $e=4$ edges and $N=e+1=5$ vertices. These results are obtained by numerically solving (\ref{Schrod}) for the total Hamiltonian (\ref{Adiapath}) with $s(t)=\dfrac{t}{T}$, using the Runge-Kutta method. The calculations are performed using Mathematica 8.0.4 for Linux, on four Dual-Core AMD Opteron(TM) 2218 processors with 16GB of RAM.

 To achieve the results of the simulations for each graph, we numerically solved (\ref{Schrod}) for the evolution times $T^{\prime}=1,2,...,20$. For each evolution time, we calculated the total success probability as described in Sect.~\ref{impl}, using (\ref{PsPart}) and (\ref{Ps}). Then, we interpolated the results at $P_{s}=0.25$ for entries without an exact value.

\begin{table}
\caption{Results of simulation for our AQC algorithm for some primitive graphs, that need at most $L=15$ qubits. $D$ is the dimension of degeneracy space of ground states of $H_{p}$ and $T^{\prime}$ is the evolution time to reach the success probability $P_{s}=0.25$.}
{\begin{tabular}{ccccc}
\hline\hline
\hspace{0pt}\text{Graph \;\;\;\;\;} \hspace{45pt}& $e$ \hspace{45pt}& $L$ \hspace{45pt} & $D$ \hspace{45pt} & $T^{\prime}$\\
\hline\hline
$\;\;\;K_{1,1}$ \hspace{45pt} & $1$ \hspace{45pt} & $2$ \hspace{45pt} & $2$ \hspace{45pt}& $0$ \\  
\hline
$\;\;\;K_{1,2}$ \hspace{45pt} & $2$ \hspace{45pt} & $6$ \hspace{45pt} & $4$ \hspace{45pt} & $2.402$ \\ 
\hline
$\;\;\;Z_{4}$ \hspace{45pt} & \hspace{45pt} & \hspace{45pt} & $4$ \hspace{45pt} & $7.075$ \\  
$K_{1,3}$ \hspace{45pt} & $3$ \hspace{45pt} & $8$ \hspace{45pt} & $12$ \hspace{45pt} & $2.292$ \\ 
$K_{3}$ \hspace{45pt} & \hspace{45pt} & \hspace{45pt} & $12$ \hspace{45pt} & $2.504$ \\ 
\hline
$Z_{5}$ \hspace{45pt} & \hspace{45pt} & \hspace{45pt} & $8$ \hspace{45pt} & $19.360$ \\ 
$K_{1,4}$ \hspace{45pt} & \hspace{45pt} & \hspace{45pt} & $48$ \hspace{45pt} & $5.557$ \\ 
$C_{4}$ \hspace{45pt} & \hspace{45pt} & \hspace{45pt} & $16$  \hspace{45pt} & $11.221$ \\ 
$G^{4}_{1}$ \hspace{45pt} & $4$ \hspace{45pt} & $15$ \hspace{45pt} & $12$ \hspace{45pt} & $16.085$ \\ 
$G^{4}_{2}$ \hspace{45pt} & \hspace{45pt} & \hspace{45pt} & $20$ \hspace{45pt} & $9.311$ \\ 
$G^{4}_{3}$ \hspace{45pt} & \hspace{45pt} & \hspace{45pt} & $120$ \hspace{45pt} & $5.547$ \\ 
\hline\hline
\end{tabular}}\label{TRes}
\end{table}

We categorized these results in Table~\ref{TRes} according to the number of qubits ($L$). The dependence of the average of these  evolution  times (for each category) on the size of the system (number of the qubits) is shown in Fig.~\ref{ChAll} (solid line). It shows an exactly quadratic fit to the simulated data. 
Because of incorporating all possible graphs corresponding to the same $L$, this excellent fit suggests that the time complexity of the algorithm be (at least) of a polynomial order. On the other hand, if we categorize the results by known classes of graphs such as stars ($K_{1,n}$) and paths ($Z_{n}$)\footnote{According to evolution times and degeneracy degrees, these two classes can be considered as the boundary cases of the algorithm (the first is the best-case, and another is one of the worst-cases).}, we will see that an approximately quadratic fit still remains (see Fig. \ref{ChAll} dashed and dotted lines). According to our classical computer capabilities, these statements are based on the simulations up to $15$ qubits.

\begin{figure}[tbp]
\includegraphics[scale=1]{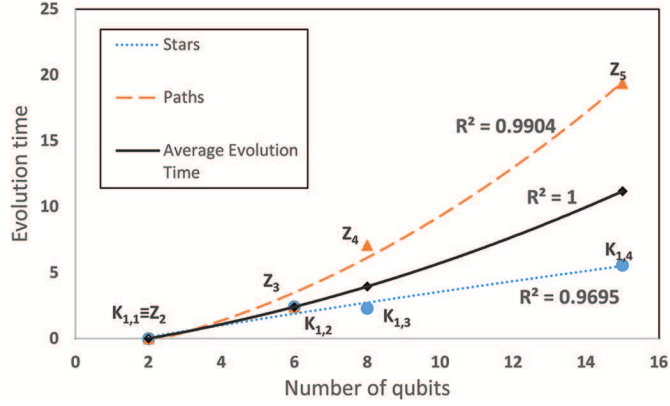} 
\vspace*{8pt}
\caption{The Average evolution time growth, and the evolution time growth for $Z_{n}$ and $K_{1,n}$ with the number of qubits, for $P_{s}=0.25$ ($R^{2}$ is the regression of the interpolated fits.}
\label{ChAll}
\end{figure}

Table~\ref{TDd} represents the degeneracy degrees of the ground states of $H_{p}$ for some larger graphs. These results are obtained by calculating $C(s_{b})$ for all possible $L$-bit strings ($L=18,21$) on a classical computer. It shows that, with a fixed number of edges, how $D$ differs for the different classes of graphs. Especially, it shows the remarkable difference between these values of $D$ for the not graceful graph $C_{5}$ and other $5$-edge graphs.

\begin{table}
\caption{Dimension of degeneracy space of the ground states of $H_{p}$ for some graphs with $e=6,7$ edges. See Table~\ref{TG} for a guide to the names of graphs.}
{\begin{tabular}{@{}cccccccccc@{}}
\hline\hline
$5$-edge graph: & $G^{5}_{1}$ & $G^{5}_{2}$ & $G^{5}_{3}$ & $G^{5}_{4}$ &  $G^{5}_{5}$ & $G^{5}_{6}$ & $K_{1,5}$ & $Z_{6}$ & $C_{5}$\\  
\hline
$D$: & $40$ & $28$ & $64$ & $72$ & $48$ & $36$ & $240$ & $24$ & $1220$\\ 
\hline\hline
$6$-edge graph: & $Z_{7}$ & $K_{1,6}$ & $G^{6}_{1}$ \\
\hline
$D$: & $32$ & $1440$ & $44$\\   
\hline\hline
\end{tabular}}\label{TDd}
\end{table}

\section{Conclusions}
\label{cncld}
In this paper, we have introduced an adiabatic quantum algorithm for graceful labelling problem for finite graphs. We also did some simulations for some simple graphs. Then we discussed that however degeneracy in the ground states of the problem Hamiltonian may decrease the chance of the system to pass a true adiabatic evolution, but it does increase the probability of finding the system in any of the desired eigenstates (or a superposition of them), after any pre-estimated evolution time $T^{\prime}$.  Simulations are carried out for systems of up to 15 qubits. Finally, we performed a detailed analysis of the simulation results, which showed that the time complexity of the algorithm can be of a polynomial order with respect to the number of qubits. 

\section*{Acknowledgments}

We should acknowledge P. Sheikholeslam and E. Najafi at MUT High Performance Computing Center, who specially dedicated their valuable time out of schedule to different steps of our computations.

\vspace*{-6pt}   


\begin{thebibliography}{0}

\bibitem{book}
A. Messiah, {\it Quantum Mechanics}, Vol.~II, (North Holland, Amsterdam, 1962), pp.~740.

\bibitem{GrovFarhi} 
E. Farhi and S. Gutmann, {\it et al}., {\it Phys. Rev. A} {\bf 57} 4, (1998) 2403.

\bibitem{Robustness}
A. M. Childs, E. Farhi and J. Preskill, {\it Phys. Rev. A} {\bf 65}, (2001) 012322 .

\bibitem{nature}
T. D. Ladd, F. Jelezko, R. Laflamme, Y. Nakamura, C. Monroe and J. L. O’Brien, {\it Nature} {\bf 464}, (2010) pp.~45--53.
 
\bibitem{ThesisF}
D. Gosset, Case studies in quantum adiabatic optimization, PhD Thesis, Department of Physics, MIT (2011).

\bibitem{QEFarhi}
E. Farhi, J. Goldstone, S. Gutmann, J. Lapan, A. Lundgren, and D. Preda, {\it Science} {\bf 292}, (2001) pp.~472.

\bibitem{QCFarhi}
E. Farhi, J. Goldstone, S. Gutmann and M. Sipser, quant-ph/0001106 (2000).

\bibitem{VanDam}
W. van Dam, M. Mosca and U. Vazirani, How powerful is adiabatic quantum computation? in {\it IEEE Conf. Proc. $42$nd IEEE Symp. Foundations of Computer Science} (2001).

\bibitem{Fail}
E. Farhi, J. Goldstone, S. Gutmann and D. Nagaj, {\it Int. J. Quantum Inform.} {\bf 6}(3),  (2003) pp.~503--518.

\bibitem{NG}
D. M. Tong, K. Singh, L. C. Kwek, and C. H. Oh, {\it Phys. Rev. Lett.} {\bf 95}(5), (2005) 110407.

\bibitem{NPCV}
N. G. Dickson and M. H. S. Amin, {\it Phys. Rev. Lett.} {\bf 106}(5), (2011) 050502.

\bibitem{2satExp}
S. Santra, G. Quiroz, G.V. Steeg and D.A. Lidar, Max 2-SAT with up to 108 qubits. {\it New J. Phys.} {\bf 16}, (2014) 045006.

\bibitem{3satExp}
M. El-fi ky, S. Ono and S. Nakayama, {\it Artif. Life Robotics} {\bf 111}, (2011) pp.~108--111.

\bibitem{IntFactor}
X. Peng, Z. Liao, N. Xu, G. Qin, X. Zhou, D. Suter, and J. Du, {\it Phys. Rev. Lett.} {\bf 101}, (2008) 220405.

\bibitem{Simon} 
M. V. P. Rao, {\it Phys. Rev. A} {\bf 67}, (2003) 052306.

\bibitem{SimonExp}
Y. Long, G. Feng, and Y. Tang, {\it Phys. Rev. A} {\bf 88}, (2013) 012306. 

\bibitem{ramsey} 
F. Gaitan and L. Clark, {\it Phys. Rev. Lett.} {\bf 108}, (2012) 010501. 

\bibitem{ramseyExp} 
Z. Bian, ,F. Chudak, ,W. G. Macready, L. Clark, and F. Gaitan, {\it Phys. Rev. Lett.} {\bf 111}, (2013) 130505. 
 
\bibitem{isomorphism} 
F. Gaitan and L. Clark, {\it Phys. Rev. A} {\bf 89}, (2014) 022342.

\bibitem{TSP} 
R. H. Warren, {\it Quantum Inf. Process.} {\bf 12}, (2013) 1781.

\bibitem{Rosa} 
A. Rosa, On certain valuations of the vertices of a graph {\it Theory of Graphs, Internat. Symposium}, Rome (1966), Gordon and Breach, New York, (1967), pp.~349--355 .

\bibitem{Golomb}
S.W. Golomb, How to number a graph, in {\it Graph Theory and Computing} ed. R. C. Read (Academic Press, New York, 1972) pp.~23--37.
  
\bibitem{Erdos}
R. L. Graham, N. J. A. Sloane, {\it SIAM J. Alg. Discrete Math.} {\bf 1}, (1980) pp.~382--404.

\bibitem{Sheppard}
D. A. Sheppard, {\it Discrete Math.} {\bf 15}, (1976) 379. 

\bibitem{Bloom}
G. S. Bloom and S. W. Golomb, in {\it Proc. of the IEEE} {\bf 65}, (1977), pp.~562--570.

\bibitem{Ringel-Kotzig}
C. Huang, A. Kotzig, and A. Rosa, {\it Util. Math.} {\bf 21c}, (1982) 31. 

\bibitem{InTree}
B. D. Acharya, S. Rao and S. B. Arumugan,  {\it Embeddings and NP-complete problems for graceful graphs, Labelling of Discrete Structures and Applications.}, (Narosa Publishing House, New Delhi, 2008) pp.~57--62.

\bibitem{stateofart}
L. Brankovic and I. M. Wanless, {\it Math.Comput.Sci.} {\bf 5}, (2011), pp.~11--20, DOI 10.1007/s11786-011-0073-6.

\bibitem{Redl}
T. A. Redl, Graceful graphs and graceful labelings: two mathematical programming formulations and some other new results. {\it Tech. Report TR$03-01$} CAAM Department, Rice University, Texas (2003).

\bibitem{Eshghi}
K. Eshghi and P. Azimi,  {\it J. Appl. Math.} {\bf 2004: 1}, (2004) pp.~1--8.
 
\bibitem{Meta}
H. Mahmoudzadeh, and K. Eshghi, in {\it Proc. of the IEEE Swarm Intelligence Symposium}, Honolulu (2007), pp.~84--91, DOI: 10.1109/SIS.2007.368030.

\bibitem{Galian}
 J. A Gallian, A dynamic survey of graph labeling {\it Electronic J. Combin.} {\bf 18}, (2014) \#DS6 
http://www.combinatorics.org/ojs/index.php/eljc/article/view/ds6.
 

\bibitem{MC}
K. A. Berman, J. L. Paul  {\it Algorithms: sequential, parallel, and distributed.}, (Thomson Course Technology, Boston 2005) pp.~757
 
%


\end{thebibliography}
\end{document}